\documentclass[11pt]{article}

\usepackage{amsthm,amsmath,amssymb}
\usepackage[bookmarks=true,hypertexnames=false,pagebackref]{hyperref}
\hypersetup{colorlinks=true, citecolor=blue, linkcolor=red,
  urlcolor=blue}
\usepackage{newpxtext}
\usepackage[libertine,vvarbb]{newtxmath}

\usepackage{fullpage}
\usepackage{hyperref}
\usepackage{cleveref}
\usepackage[utf8]{inputenc}
\usepackage{framed}
\usepackage{enumerate}
\usepackage{graphicx}
\usepackage{float} 
\usepackage{subfigure}

\usepackage[ruled,vlined]{algorithm2e}
\usepackage{algorithmic}

\usepackage{tikz}
\usepackage{complexity}
\usepackage{ulem}

\newtheorem{theorem}{Theorem}[section]
\newtheorem{claim}[theorem]{Claim}

\newtheorem{lemma}[theorem]{Lemma}
\newtheorem{definition}[theorem]{Definition}
\newtheorem{conjecture}[theorem]{Conjecture}

\newcommand{\ent}{\mathcal{H}}

\newcommand{\gip}{\mathrm{GIP}}

\newcommand{\GHM}{\text{HM}}
\newcommand{\rv}[1]{\boldsymbol{#1}}

\title{Exponential Separation of Quantum and Classical One-Way Numbers-on-Forehead Communication}

\begin{document}

\author{
Guangxu Yang \thanks{Research supported by NSF CAREER award 2141536.\\ Thomas Lord Department of Computer Science, University of Southern California.  \\ Email: \texttt{\{guangxuy\}@usc.edu}} \\
\and
Jiapeng Zhang \footnotemark[1]
}

\maketitle
\begin{abstract}
Numbers-on-Forehead (NOF) communication model is a central model in communication complexity. As a restricted variant, one-way NOF model is of particular interest. Establishing strong one-way NOF lower bounds would imply circuit lower bounds, resolve well-known problems in additive combinatorics, and yield wide-ranging applications in areas such as cryptography and distributed computing. However, proving strong lower bounds in one-way NOF communication remains highly challenging; many fundamental questions in one-way NOF communication remain wide open. One of the fundamental questions, proposed by Gavinsky and Pudlák (CCC 2008), is to establish an explicit exponential separation between quantum and classical one-way NOF communication.

In this paper, we resolve this open problem by establishing the first exponential separation between quantum and randomized communication complexity in one-way NOF model. Specifically, we define a lifted variant of the Hidden Matching problem of Bar-Yossef, Jayram, and Kerenidis (STOC 2004) and show that it admits an ($O(\log n)$)-cost quantum protocol in the one-way NOF setting. By contrast, we prove that any $k$-party one-way randomized protocol for this problem requires communication $\Omega(\frac{n^{1/3}}{2^{k/3}})$. Notably, our separation applies even to a generalization of $k$-player one-way communication, where the first player speaks once, and all other $k-1$ players can communicate freely.  

\end{abstract}
\section{Introduction}
A fundamental objective in the study of quantum advantage is to exhibit computational separations across diverse models. Communication complexity has seen substantial progress toward this goal: a line of work ~\cite{buhrman1998quantum,raz1999exponential,bar2004exponential,gavinsky2007exponential,regev2011quantum,gavinsky2016entangled,girish2022quantum,gavinsky2019quantum,gavinsky2020bare,goos2024quantum} now establishes exponential separations in several two-party communication settings. Building on these foundations, Aaronson, Buhrman, and Kretschmer recently introduced \textit{quantum information supremacy}~\cite{ABK24}—a paradigm focusing on tasks where a quantum device utilizes exponentially fewer resources (qubits) than any classical counterpart (bits). Their proposed task, rooted in the classical–quantum separation of the \textit{Hidden Matching} problem, was recently realized in a landmark trapped-ion experiment~\cite{KGDGGGHMNHA25}.

Despite these successes in two-party communication complexity, quantum advantage in the multiparty setting remains poorly understood. The premier model for multiparty communication is the Number-on-the-Forehead (NOF) model, introduced by Chandra, Furst, and Lipton~\cite{chandra1983multi}. In this model, the input is partitioned into $k$ parts, and player $i$ can see all parts except their own (as if it were "written on their forehead"). While the $k = 2$ case coincides with the standard two-party model, for $k \ge 3$, NOF protocols can simulate a strictly broader range of computational models, making it a uniquely powerful and challenging framework. Recently, Göös, Gur, Jain, and Li~\cite{goos2024quantum} highlighted the separation of quantum and randomized communication in NOF models as a major open problem.

The main obstacle to obtaining strong separations between quantum and randomized protocols in the NOF model, already noted in \cite{goos2024quantum,yang2025quantum}, is the lack of techniques for proving lower bounds for deterministic and randomized NOF protocols. Most existing methods, such as the discrepancy method \cite{rao2015simplified}, apply equally well to quantum, randomized, and deterministic communication, making it difficult to distinguish the power of quantum versus classical protocols in the NOF setting~\cite{lee2009lower}. 

Therefore, it is natural to consider the quantum advantage of the NOF model in restricted settings. Similar to the research on two-party CC, the focus typically shifts from simultaneous models to one-way and then to two-way communication \cite{buhrman1998quantum,raz1999exponential,bar2004exponential,gavinsky2007exponential,regev2011quantum,gavinsky2016entangled,girish2022quantum,gavinsky2019quantum,gavinsky2020bare,goos2024quantum}. A remarkably natural and well-studied restriction of this model is the one-way NOF communication model, in which each player is allowed to speak exactly once according to a fixed order. Some progress has been made on the simultaneous NOF communication model, which is a weaker variant than the one-way NOF model. Gavinsky and Pudlák~\cite{QCCinNOF} established the first $O(\log n)$ vs $\Omega(n^{1/6k^2})$ separation between quantum and classical communication in the simultaneous NOF communication setting. Recently, Yang and Zhang~\cite{yang2025quantum} improved this to an $O(\log n)$ vs $\Omega(n^{1/16})$ separation for three-player simultaneous NOF. 
In \cite{QCCinNOF} and \cite{yang2025quantum}, they further posed the following open problem:
\begin{center}
\textit{Can exponential quantum advantage be established in one-way Numbers-on-Forehead communication?}
\end{center}
As mentioned in \cite{QCCinNOF}, even the three-player case of one-way Numbers-on-Forehead communication is highly nontrivial and of independent interest.

\paragraph{Why we study one-way NOF communication.}
Although one-way NOF model may at first appear to be a restricted variant of the general NOF model, it nonetheless possesses significant theoretical motivations and a wide range of applications.

First, the model holds pivotal importance for several long-standing open problems in circuit complexity. For instance, proving an $\omega(\log n)$ lower bound for any function $f$ in $k$-party deterministic one-way NOF communication with $k = \log n$ would imply that $f \notin \mathrm{ACC}^0$ \cite{HG90, beigel1994acc, chakrabarti2007lower, VW07}. Even more strikingly, establishing strong deterministic one-way lower bounds for only three players in specific problems would yield significant size-depth trade-offs for Boolean circuits \cite{valiant1977graph, pudlak1997boolean}. Beyond circuit complexity, one-way NOF lower bounds have deep connections with additive combinatorics. Remarkably, establishing strong deterministic one-way NOF lower bounds for specific problems yields quantitative bounds for the Hales-Jewett theorem, dense Ruzsa-Szemer'{e}di graphs, $k$-term arithmetic-progression-free sets~\cite{linial2017communication,jaber2025}.

More broadly, one-way NOF communication model has wide applications in theoretical computer science, with particularly profound implications for cryptography. In cryptography, one-way NOF communication is intrinsically linked to protocols of Private Information Retrieval (PIR)~\cite{chor1998private}, lower bounds in position-based cryptography, lower bounds in function inversion~\cite{brody2017position,corrigan2019function}, and leakage-resilient key exchange protocols~\cite{li2020leakage}.  Beyond these cryptographic primitives, it's further evidenced by its direct applications in lower bounds in distributed computing~\cite{drucker2014power}, construction of space-bounded pseudorandom generators~\cite{babai1992multiparty,ganor2014space},time-space trade-off of oblivious branching programs~\cite{VW07}, and lower bounds of streaming algorithms~\cite{kallaugher2019complexity}.

Despite its profound implications and extensive applications, several fundamental questions regarding the one-way NOF communication model remain unresolved. Most notably, as recently highlighted by \cite{QCCinNOF} and \cite{yang2025quantum}, establishing an explicit exponential separation between randomized and quantum one-way NOF communication remains a major open problem. 
Establishing strong lower bounds in one-way NOF communication remains highly challenging. In addition to the aforementioned discrepancy method, several other techniques have been developed for proving one-way NOF lower bounds ~\cite{pudlak1997boolean, beame2005direct, VW07}. However, none of these approaches can surpass the $\Omega(n^{1/(k - 1)})$ lower bound barrier.

\paragraph{Quantum advantage via lifting techniques}
Lifting theorems are a generic method for translating lower bounds from weaker computational models to relatively stronger ones. A representative example of lifting theorems is the query-to-communication lifting theorems \cite{raz1997separation, zhang2009tightness, goos2018deterministic,pitassi2017strongly,goos2020query,chattopadhyay2019query,lovett2020lifting, Collision,mao2025gadgetless}, which convert lower bounds in query complexity into communication complexity lower bounds, using a suitable base function composed with a gadget.

In the two-party setting, applying lifting to problems where quantum decision trees offer an exponential advantage over randomized ones, such as the Forrelation problem \cite{aaronson2010bqp, aaronson2015forrelation}, we can get the exponential separations between two-way quantum and classical communication complexity \cite{girish2022quantum,goos2024quantum}. Investigating whether similar lifting techniques can be extended to the more complex one-way NOF setting is a natural and highly promising direction.

Recently, Yang and Zhang \cite{yang2025deterministic} proposed a novel lifting framework connecting two-party communication with the Number-on-the-Forehead (NOF) model. Specifically, they established a deterministic lifting theorem that translates one-way two-party lower bounds into the one-way NOF setting. However, they left the randomized version as a significant open problem, achieving only deterministic bounds at present. 

Inspired by the technique in \cite{yang2025deterministic}, we successfully bypass this $\Omega(n^{1/(k - 1)})$ lower bound barrier and resolve the aforementioned open problem. Specifically, we establish the first explicit exponential separation between randomized and quantum one-way NOF communication by using lifting techniques.


\subsection{Our Contributions}

In this paper, we establish the first exponential separation between quantum and randomized communication complexity in the one-way NOF model.

Our construction is based on a lifted version of the Hidden Matching (HM) problem introduced by Bar-Yossef et al.~\cite{bar2004exponential}. In the standard two-party HM problem:\begin{enumerate}\item Alice is given a string $z \in \{0,1\}^n$.\item Bob is given a perfect matching $M \in \mathcal{M}_n$ on $n$ nodes.\end{enumerate}The goal is for Bob to output a tuple $\langle i, j, b \rangle$ such that the edge $(i, j) \in M$ and $b = z_i \oplus z_j$. 

It was shown in~\cite{bar2004exponential} that HM admits an $O(\log n)$ simultaneous quantum protocol, whereas any one-way randomized protocol requires $\Omega(\sqrt{n})$ communication.

In our construction, we utilize a structured subset of perfect matchings. Fix $m$, for each $i \in [m]$, we define a perfect matching $M_i$ between the sets $\{0, \dots, m-1\}$ and $\{m, \dots, 2m-1\}$ as:$$M_i := \left\{ \left(\ell, m + ((i + \ell) \bmod m)\right) : \ell \in [m] \right\}.$$Let $\mathcal{M} = \{M_1, \dots, M_m\}$ denote this collection. The Lifted Hidden Matching problem, denoted by $\text{HM} * g$, is defined as follows:
\begin{definition}\label{dfn: GHM}
Let \( g : \{0,1\}^{n(k-1)} \rightarrow \{0,1\}^{n_0}  \) be the gadget function with $n_0=(\frac{n}{2^k})^{2/3}$. The \textit{Lifted Hidden Matching Problem}, denoted \( \GHM * g \), involves $k$ inputs distributed among the players as follows: Let $m=n_0/2$,
\begin{itemize}
    \item The $i$-th player receives \( (x_1,...,x_{i-1},x_{i+1},...,x_k) \) where $x_1\in [m]$ and $x_i \in \{0,1\}^n$ for $2\leq i\leq k$.
\end{itemize}
The goal is for the last player to output a tuple \( \langle \ell, r, b \rangle \) such that \( (\ell, r) \in M_{x_1} \) and \( b = g(x_2,...,x_k)_\ell \oplus g(x_2,...,x_k)_{r} \).
\end{definition}

Remarkably, this problem remains easy for quantum communication. Similar to the Hidden Matching Problem, the first player only needs to send a uniform superposition of the string \( g(x_2,...,x_k) \), with a communication cost of \( O(\log n) \) qubits. The last player can then perform a measurement on this superposition, which depends on the matching \( M_{x_1} \), and output the parity of some pair in \( M_{x_1} \) (see the appendix for more details). We note that other players do not need to send any message using this protocol. Thus, the main result of our paper is the \( n^{\Omega(1)} \) randomized one-way NOF communication lower bound.

\begin{theorem}\label{main_theorem}
Let $n_0=(\frac{n}{2^k})^{2/3}$, there exists an explicit gadget function \( g : \{0,1\}^{n(k-1)} \rightarrow \{0,1\}^{n_0}\) such that the randomized one-way NOF communication complexity of \( \GHM * g \) is $\Omega(\frac{n^{1/3}}{2^{k/3}})$.
\end{theorem}

\section{Preliminaries}
\subsection{Communication Complexity}We begin by recalling standard definitions in communication complexity. In the two-party communication model, Alice and Bob receive inputs $x \in X$ and $y \in Y$, respectively. Rather than computing a simple boolean function, their goal is to solve a \textit{search problem} (or compute a relation) $\mathcal{P} \subseteq X \times Y \times Z$. For any given input $(x,y)$, the objective is to output a valid answer $z \in Z$ such that $(x,y,z) \in \mathcal{P}$.

\begin{definition}[Randomized one-way communication complexity]
In the randomized one-way communication model, Alice and Bob share a public random string $r$. Alice sends a single message $\Pi(x, r)$ to Bob, and Bob outputs an answer $z \in Z$ based on $y$, the received message, and the shared randomness. A protocol has $\delta$-error if, for every input $(x, y)$, the probability over the shared randomness that Bob outputs an invalid answer is at most $\delta$. The randomized one-way communication complexity of $\mathcal{P}$ with error $\delta$ is the maximum length of Alice's message over all inputs and random strings.
\end{definition}

By Yao's Minimax Principle, to prove a lower bound for randomized protocols with error $\delta$ in the worst case, it is sufficient to prove a lower bound for the \textit{distributional complexity}—that is, the communication complexity of any deterministic protocol that errs with probability at most $\delta$ under a chosen input distribution (in our case, the uniform distribution).

\begin{definition}[One-way NOF Communication Model for Search Problems]
In the $k$-party one-way Numbers-on-Forehead (NOF) model, $k$ players collaborate to solve a search problem $\mathcal{P} \subseteq X_1 \times \cdots \times X_k \times Z$. The inputs are distributed such that each player $i$ knows all inputs except for their own $x_i$ (i.e., player $i$ sees $x_{-i}$).In the randomized one-way setting, the players communicate in a fixed order, from the first player to the last. Each player $i$ sends a single message $\Pi_i(x_{-i}, r)$ based on their visible input, the previously sent messages, and the shared randomness $r$. The last player then outputs an answer $z \in Z$.
\end{definition}

\begin{definition}[Cylinder Intersections]A set $S \subseteq X_1 \times \cdots \times X_k$ is called a cylinder if there exists an index $i \in [k]$ such that membership in $S$ does not depend on the value of $x_i$. A set $S$ is called a cylinder intersection if it can be written as $S = S_1 \cap \cdots \cap S_k$, where each $S_i$ is a cylinder. 
\end{definition}

\subsection{Basics of Information Theory}\label{sec:ic}
Our proof approach involves several standard definitions and results from information theory, which we now recall.

\begin{definition}[Entropy]
Given a random variable \(\boldsymbol{X}\), the Shannon entropy of \(\boldsymbol{X}\) is defined by 
\[
\ent(\boldsymbol{X}) := \sum_{x} \Pr(\boldsymbol{X} = x) \log \left( \frac{1}{\Pr(\boldsymbol{X} = x)} \right).
\]
For two random variables \(\boldsymbol{X}, \boldsymbol{Y}\), the \textit{conditional entropy} of \(\boldsymbol{X}\) given \(\boldsymbol{Y}\) is defined by 
\[
\ent(\boldsymbol{X} ~|~ \boldsymbol{Y}) := \mathbb{E}_{y \sim \boldsymbol{Y}} \left[ \ent(\boldsymbol{X} \mid \boldsymbol{Y} = y) \right].
\]
\end{definition}

\noindent For \(p \in [0,1]\), the binary entropy function is defined as $H_2(p) = - p \log_2 p - (1 - p) \log_2 (1 - p).$ It is well known that \(H_2(p)\) is a concave function.

\begin{lemma}[Subadditivity of Entropy]\label{subadditivity}
For a list of random variables \(\boldsymbol{X}_1, \boldsymbol{X}_2, \ldots, \boldsymbol{X}_d\), we have:
\[
\ent(\boldsymbol{X}_1, \boldsymbol{X}_2, \ldots, \boldsymbol{X}_d) \leq \sum_{i=1}^{d} \ent(\boldsymbol{X}_i).
\]
\end{lemma}

\begin{definition}[Mutual Information]
The mutual information between joint random variables $\boldsymbol{X}$ and $\boldsymbol{Y}$ is defined as
\[
\mathrm{I}(\boldsymbol{X}; \boldsymbol{Y}) = \ent(\boldsymbol{X}) - \ent(\boldsymbol{X} | \boldsymbol{Y}),
\]
\end{definition}

\begin{lemma}[Data Processing Inequality]\label{data_process}
Consider random variables \(\rv{X}, \rv{Y}, \rv{Z}\) forming a Markov chain \(\rv{X} \rightarrow \rv{Y} \rightarrow \rv{Z}\). Then, the mutual information satisfies:
\[
\mathrm{I}(X; Y) \geq \mathrm{I}(X; Z).
\]
\end{lemma}

\begin{definition}[Hamming Distance]
Let $x = (x_1, \ldots, x_n), y = (y_1, \ldots, y_n )\in \{0,1\}^n$ be two strings. Their Hamming distance \(d_H(x, y)\) is defined as:
\[
d_H(x, y) := |\{i: x_i\neq y_i\}|.
\]
\end{definition}

\begin{definition}[Total Variation Distance]
For two discrete random variables $\boldsymbol{X}$ and $\boldsymbol{Y}$ taking values in a finite sample space $\Omega$, the total variation distance between their probability distributions is defined as:$$\Delta_{TV}(\boldsymbol{X}, \boldsymbol{Y}) := \frac{1}{2} \sum_{\omega \in \Omega} |\Pr(\boldsymbol{X} = \omega) - \Pr(\boldsymbol{Y} = \omega)|.$$In this paper, we use $U_m$ to denote a random variable that is uniformly distributed over the Boolean hypercube $\{0,1\}^m$. For instance, $\Delta_{TV}(\boldsymbol{Z}, U_{n_0})$ measures the statistical distance between the distribution of a random variable $\boldsymbol{Z} \in \{0,1\}^{n_0}$ and the uniform distribution over $\{0,1\}^{n_0}$.
\end{definition}

\section{Quantum Advantage for One-Way NOF Communication}

We prove the exponential separation between quantum and randomized one-way Numbers-on-Forehead (NOF) communication in this section. We first recall the statement.

\begin{theorem}[Theorem \ref{main_theorem} restated]
There exists an explicit gadget function \( g : \{0,1\}^{n(k-1)} \rightarrow \{0,1\}^{n_0}\) such that the randomized one-way NOF communication complexity of \( \GHM * g \) is $\Omega\left(\frac{n^{1/3}}{2^{k/3}}\right)$. 
\end{theorem}

\subsection{Step 1: Choosing the gadget function}
\begin{definition}
Let $q$ be a prime power and $k, r > 0$. We define the function $\gip_{q,r}^{k}: (\mathbb{F}_{q}^{r})^{k} \rightarrow \mathbb{F}_q$ by
\[
\gip_{q,r}^{k}(x_1,\dots,x_k) = \sum_{j \in [r]} \prod_{i \in [k]} x_{i,j},
\]
where $\mathbb{F}_q$ is a finite field, and all arithmetic operations are over $\mathbb{F}_q$. When $q, r, k$ are clear from context, we write $\gip(x_1,\dots,x_k)$ for simplicity.
\end{definition}

We note that $\gip$ is a cylinder intersection extractor:
\begin{lemma} \label{disperser} \cite{yang2025deterministic}
For $r\geq 2^{k+1}.$ Let $S\subseteq (\mathbb{F}_{q}^{r})^{k}$ be any cylinder intersection of size $|S|\geq q^{r\cdot k-1}$. Then for every $v\in\mathbb{F}_{q}$,  we have that
\[
\Pr_{(x_1,\dots,x_k) \in S}\left[\gip(x_1,\dots,x_k) = v\right] \leq \frac{1}{q} + q\cdot (k/q)^{4}
\]
\end{lemma}

We modify the function $\gip_{q,r}^{k-1}$ by setting the field size $q=2^{n/2^k}$ and the number of blocks $r=2^{k+1}$. We define our gadget function $g: \{0,1\}^{n(k-1)} \rightarrow \{0,1\}^{n_0}$ such that its output is represented by the first $n_0$ bits of $\gip_{q,r}^{k-1}$.

\begin{lemma} \label{extractor}
Let $S \subseteq \{0,1\}^{n(k-1)}$ be any cylinder intersection of size $|S| \geq 2^{n(k-1)- n/2^{k}}$. Then for any $z \in \{0,1\}^{n_0}$,
$$ \Pr_{(x_2,\dots,x_k) \in S}\left[g(x_2,\dots,x_k) = z\right] \leq 2^{-n_0} + 2^{-2n_0}. $$
\end{lemma}

\begin{proof}
Let $N = n/2^k$, so the field size is $q = 2^N$. We identify $\mathbb{F}_q$ with the Boolean hypercube $\{0,1\}^N$. By the discrepancy bound for the $\gip$ function over the cylinder intersection $S$, the distribution of the full output of $\gip$ is extremely close to uniform. Specifically, for any $y \in \{0,1\}^N$, the point-wise probability is bounded by:
$$ \Pr_{(x_2,\dots,x_k) \in S}[\gip(x_2,\dots,x_k) = y] \leq 2^{-N} + 2^{-2N}, $$
based on the exponential sum bound for our choice of $r = 2^{k+1}$ and $k \leq \frac{\log n}{2}$.

Our gadget function $Z = g(x_2, \dots, x_k)$ outputs exactly the first $n_0$ bits of $\gip$. For any specific prefix $z \in \{0,1\}^{n_0}$, there are exactly $2^{N - n_0}$ full outputs $y \in \{0,1\}^N$ that match this prefix. Therefore, the probability of obtaining $z$ over the cylinder intersection $S$ is the sum of the probabilities of these matching extensions:
\begin{align*}
\Pr_{(x_2,\dots,x_k) \in S}[g(x_2,\dots,x_k) = z] &= \sum_{y \in \{0,1\}^N: y \text{ matches } z} \Pr_{(x_2,\dots,x_k) \in S}[\gip(x_2,\dots,x_k) = y] \\
&\leq 2^{N - n_0} \left( 2^{-N} + 2^{-2N} \right) \\
&= 2^{-n_0} + 2^{-N - n_0} \leq 2^{-n_0} + 2^{-2n_0}.
\end{align*}
Where the last inequality follows by the fact that $n_0 = (n/2^k)^{2/3} \leq n/2^k = N$.
\end{proof}

\subsection{Step 2: Simplifying one-way NOF protocols}\label{symmetric}
To prove the randomized communication lower bound, by Yao’s minimax principle, it is sufficient to prove a lower bound for any protocol that achieves a constant advantage under a uniform input distribution.

To simplify our analysis, we first apply the following simplified process to any one-way NOF protocol. Let $m = n_0/2$ denote the number of possible matchings.

\begin{definition}[Simplified protocols]
For any protocol $\Pi$ for $\GHM * g$ under the uniform distribution, we define its simplified version $\Pi^*$ as follows:
\begin{itemize}
    \item The first player sends $\Pi^{*}_1(x_{-1}) = \Pi_1(x_{-1})$.
    \item For each $2\leq i \leq k$, the $i$-th player sends $\Pi^{*}_i(x_{-i}) = (\Pi_i(j,x_2,...,x_{i-1},x_{i+1},...,x_k))_{j\in [m]}$.
\end{itemize}
\end{definition}

The high-level intuition behind the simplified protocols is that each player $i$ with $2\leq i \leq k$ enumerates its transcript for all possible values of $x_1 \in [m]$ and sends the entire list. We observe the following useful property of simplified protocols.

\begin{claim}\label{symmetric_property}
Let $\Pi$ be a deterministic protocol for $\GHM * g$ with $\delta$-error under the uniform distribution. Then the following statements hold:
\begin{enumerate}
    \item\label{1} The protocol $\Pi^*$ is a deterministic protocol for $\GHM * g$ with $\delta$-error under the uniform distribution.
    \item\label{2} $|\Pi^{*}_1| = |\Pi_1|$ and $|\Pi^{*}_i| = m \cdot |\Pi_i|$ for $2\leq i \leq k$.  The communication complexity of $\Pi^{*}$ is $|\Pi_1| + m \cdot \sum_{i=2}^k |\Pi_i|$.
    \item\label{3} Under the uniform distribution, the protocol $\Pi^*$ depend only on $(x_2,...,x_k)$. That is, they are independent of $x_1$, respectively.
\end{enumerate}
\end{claim}
\begin{proof}
The claim follows directly from the construction of $\Pi^*$. For item \ref{1}, the last player knows the actual value of $x_1$ (due to the NOF model) and can simply extract the $x_1$-th component from the tuple sent by each player $i \geq 2$. This perfectly simulates the original protocol $\Pi$, thus achieving the exact same $\delta$-error. Item \ref{2} is immediate since each player $i \geq 2$ sends exactly $m$ simulated transcripts of $\Pi_i$. Item \ref{3} holds because Player 1 does not see $x_1$ in the NOF model, and for $i \geq 2$, enumerating the messages over all possible dummy values $j \in [m]$ completely removes the dependence on the actual input $x_1$.
\end{proof}

We emphasize that item \ref{3} is crucial and significantly simplifies our subsequent analysis. By ensuring that the entire communication transcript of $\Pi^*$ is strictly independent of $x_1$, we can evaluate how much information the protocol reveals about the gadget output $g(x_2, \dots, x_k)$ without needing to condition on the specific matching $M_{x_1}$.

\subsection{Step 3: Proving communication lower bounds for simplified protocols}

To analyze the communication complexity of simplified protocols, we begin with a key observation: since the protocol $\Pi^*$ is independent of $x_1$, it inherently lacks information regarding the matching $M_{x_1}$. Consequently, to solve the Hidden Matching problem, $\Pi^*$ must contain sufficient information about $g(x_2, \dots, x_k)$. The intuition of our analysis is to consider two exhaustive cases based on how the information about $g$ is distributed among the players:

\begin{itemize}

\item \textbf{Case 1:} The first player's message $\Pi^{*}_1(x_2,\ldots,x_k)$ contains enough information about $g(x_2,\ldots,x_k)$ such that the last player can compute $\GHM*g$. In this case, we must have $|\Pi^{*}_1| = \Omega(\sqrt{n_0})$, following the known randomized one-way communication complexity of the \textit{Hidden Matching} problem \cite{bar2004exponential}.
\item \textbf{Case 2:} The remaining $k-1$ players collectively communicate sufficient information about $g(x_2,\ldots,x_k)$ through NOF communication. In this case, the total communication must satisfy $\sum_{i=2}^k |\Pi_i^{*}| = \Omega(n/2^k)$. This follows from the \textit{cylinder intersection extractor} properties of $g$, which dictate that the distribution of $g(x_2, \dots, x_k)$ remains statistically close to uniform unless the communication is $\Omega(n/2^k)$.
\end{itemize}

This intuition is formalized in the following theorem:

\begin{theorem}\label{Information_complexity}
Let $g:\{0,1\}^{n(k-1)} \rightarrow \{0,1\}^{n_0}$ be the gadget function described in Lemma \ref{extractor}. For any simplified protocol $\Pi^*$ for $\GHM * g$ under the uniform distribution with error $1/16$, we have
$$|\Pi^{*}_1| = \Omega(\sqrt{n_0}) \quad \text{or} \quad \sum_{i=2}^k |\Pi_i^{*}| = \Omega(n/2^k).$$
\end{theorem}

We observe that our main result (Theorem~\ref{main_theorem}) follows as a direct consequence of Theorem~\ref{Information_complexity}. Setting the parameters $n_0 = \left(\frac{n}{2^k}\right)^{2/3}$ and $m = n_0/2$, by Theorem~\ref{Information_complexity}  and Claim \ref{symmetric_property}, for any  protocol $\Pi$ for $\GHM * g$ under the uniform distribution with error $1/16$:
$$|\Pi| = \Omega\left(|\Pi_1^*| + \frac{1}{m} \sum_{i=2}^k |\Pi_i^{*}| \right) = \Omega\left(\frac{n^{1/3}}{2^{k/3}}\right).$$
This completes the proof of Theorem~\ref{main_theorem}.

We now dedicate the remainder of this section to establishing Theorem \ref{Information_complexity}. Let $\Pi^{*}$ be any simplified protocol that solves $\GHM * g$ with an overall expected error at most $1/16$. Let $\boldsymbol{\tau}$ denote the random variable representing the full transcript of the protocol $\Pi^*$. Let the inputs $(\rv{x_2}, \dots, \rv{x_k})$ be random variables uniformly distributed over $\{0,1\}^{(k-1)n}$, and let $Z = g(\rv{x_2}, \dots, \rv{x_k}) \in \{0,1\}^{n_0}$ be the induced hidden target random variable.

\begin{proof}[Proof of Theorem \ref{Information_complexity}]
We assume the communication cost of players 2 through $k$ is bounded by $c_0 = \sum_{i=2}^k |\Pi_i^{*}| = o(n/2^k)$. Let $c_1 = |\Pi^{*}_1|$ denote the communication cost of the first player. We aim to prove that $c_1 = \Omega(\sqrt{n_0})$.

The proof proceeds in two main parts. In the first part, we leverage the condition $c_0 = o(n/2^k)$ alongside the properties of our cylinder intersection extractor to establish an information upper bound of $\mathrm{I}(Z; \boldsymbol{\tau}) \leq c_1 + 2$, which we formalize as Lemma \ref{info_upper_bound}. In the second part, we utilize the structural properties of the one-way Hidden Matching problem and simplified protocols to prove a matching lower bound of $\mathrm{I}(Z; \boldsymbol{\tau}) = \Omega(\sqrt{n_0})$, which we formalize as Lemma \ref{entropy_loss}. 

By combining Lemma \ref{info_upper_bound} and Lemma \ref{entropy_loss}, we immediately conclude that:
$$ c_1 + 2 \geq \mathrm{I}(Z ; \boldsymbol{\tau}) \geq \Omega(\sqrt{n_0}), $$
which directly implies $|\Pi_1^*| = c_1 = \Omega(\sqrt{n_0})$. This completes the proof of Theorem \ref{Information_complexity}.
\end{proof}


We first prove the information upper bound.

\begin{lemma}\label{info_upper_bound}
Assuming the communication cost of players 2 through $k$ in $\Pi^*$ is $c_0 = o(n/2^k)$, then $\mathrm{I}(Z ; \boldsymbol{\tau}) \leq c_1 + 2$, where $c_1 = |\Pi^{*}_1|$.
\end{lemma}

\begin{proof}[Proof of Lemma \ref{info_upper_bound}]
Let $\mathcal{X} = X_2 \times \dots \times X_k$ be the input space visible to the first player. By Claim \ref{symmetric_property}, the entire protocol $\Pi^*$ is independent of $x_1$. Thus, the transcript random variable $\boldsymbol{\tau} = (\boldsymbol{\tau}_1, \boldsymbol{\tau}_{>1})$ is solely determined by inputs uniformly drawn from $\mathcal{X}$. For any specific realization, the first player's message $\tau_1$ confines the inputs to a subset $S_{\tau_1} \subseteq \mathcal{X}$, where there are at most $2^{c_1}$ such subsets. Given $\tau_1$, the subsequent messages $\tau_{>1}$ partition $\mathcal{X}$ into at most $2^{c_0}$ cylinder intersections. Thus, any specific transcript $\tau = (\tau_1, \tau_{>1})$ is associated with a cylinder intersection $C_\tau$. The actual set of valid inputs generating $\tau$ is exactly $S_{\tau_1} \cap C_\tau$.

Let $p_\tau = \Pr[\boldsymbol{\tau} = \tau] = \frac{|S_{\tau_1} \cap C_\tau|}{|\mathcal{X}|}$ be the true probability of generating $\tau$, and let $q_\tau = \frac{|C_\tau|}{|\mathcal{X}|}$ be the volume of its corresponding cylinder intersection. Since $S_{\tau_1} \cap C_\tau \subseteq C_\tau$, we trivially have $p_\tau \leq q_\tau$. Note that for any fixed $\tau_1$, the cylinders $\{C_\tau\}_{\tau_{>1}}$ partition $\mathcal{X}$, meaning $\sum_{\tau_{>1}} q_\tau = 1$. Consequently, the sum over all possible transcript realizations is tightly bounded:
$$ \sum_\tau q_\tau = \sum_{\tau_1} \sum_{\tau_{>1}} q_\tau \leq 2^{c_1} \times 1 = 2^{c_1}. $$

We classify a transcript $\tau$ as ''bad'' if its cylinder is small, i.e., $q_\tau \leq 2^{-n/2^k}$. Since there are at most $2^{c_0}$ cylinders for any fixed $\tau_1$, the total measure of bad cylinders given $\tau_1$ is at most $2^{c_0} \cdot 2^{-n/2^k} = 2^{c_0 - n/2^k}$. The overall probability that the protocol generates a bad transcript is bounded by summing $q_\tau$ over all bad $\tau$:
$$ \Pr[\boldsymbol{\tau} \in \text{Bad}] = \sum_{\tau \in \text{Bad}} p_\tau \leq \sum_{\tau \in \text{Bad}} q_\tau \leq \sum_{\tau_1} 2^{c_0 - n/2^k} \leq 2^{c_1 + c_0 - n/2^k}. $$
We can safely assume $c_1 \leq \frac{1}{2} \Omega(n/2^k)$ (otherwise the desired bound $c_1 = \Omega(\sqrt{n_0})$ holds trivially). Given $c_0 = o(n/2^k)$, we have $c_1 + c_0 - n/2^k \leq -\Omega(n/2^k)$. Thus, $\Pr[\boldsymbol{\tau} \in \text{Bad}] \leq 2^{-\Omega(n/2^k)}$. With overwhelming probability $1 - 2^{-\Omega(n/2^k)}$, the transcript $\tau$ is ''good''.

For any valid good transcript $\tau$ (where $p_\tau > 0$ and $q_\tau > 2^{-n/2^k}$), $C_\tau$ is a large cylinder intersection. By Lemma \ref{extractor}, the target random variable $Z = g(x_2, \dots, x_k)$ is extremely uniform over $C_\tau$, bounded by $\Pr[Z=z \mid C_\tau] \leq 2^{-n_0} + 2^{-2n_0} \leq 2^{-n_0+1}$. Conditioned on the realized transcript $\boldsymbol{\tau} = \tau$, the probability of $Z$ scales by the density. Let $x = (x_2, \dots, x_k) \in \mathcal{X}$, we have:
\begin{align*}
\Pr[Z=z \mid \boldsymbol{\tau}=\tau] 
&= \frac{|\{ x \in S_{\tau_1} \cap C_\tau : g(x) = z \}|}{|S_{\tau_1} \cap C_\tau|} \\
&\leq \frac{|\{ x \in C_\tau : g(x) = z \}|}{|S_{\tau_1} \cap C_\tau|} \\
&= \frac{\Pr[Z=z \mid C_\tau] \cdot |C_\tau|}{|S_{\tau_1} \cap C_\tau|} \\
&= \frac{\Pr[Z=z \mid C_\tau] \cdot q_\tau}{p_\tau} \leq \frac{q_\tau}{p_\tau} 2^{-n_0+1}.
\end{align*}
Taking the logarithm, the conditional entropy given a specific good transcript is lower-bounded by:
$$ \ent(Z \mid \boldsymbol{\tau}=\tau) \geq n_0 - 1 - \log \frac{q_\tau}{p_\tau}. $$

We now calculate the expected conditional entropy over the random variable $\boldsymbol{\tau}$. Since entropy is non-negative, we can simply drop the contribution from bad transcripts:
$$ \ent(Z \mid \boldsymbol{\tau}) = \sum_{\tau} p_\tau \ent(Z \mid \boldsymbol{\tau}=\tau) \geq \sum_{\tau \in \text{Good}} p_\tau \left( n_0 - 1 - \log \frac{q_\tau}{p_\tau} \right). $$
By applying Jensen's inequality to the logarithmic term over the good transcripts:
$$ \sum_{\tau \in \text{Good}} p_\tau \log \frac{q_\tau}{p_\tau} \leq \log \left( \sum_{\tau \in \text{Good}} p_\tau \frac{q_\tau}{p_\tau} \right) \leq \log \left( \sum_{\tau} q_\tau \right) \leq \log(2^{c_1}) = c_1. $$
Thus, the expected conditional entropy is firmly bounded by:
$$ \ent(Z \mid \boldsymbol{\tau}) \geq \Pr[\boldsymbol{\tau} \in \text{Good}] (n_0 - 1) - c_1 \geq \left(1 - 2^{-\Omega(n/2^k)}\right)(n_0 - 1) - c_1 \geq n_0 - c_1 - 1 - o(1). $$

Since $Z$ is uniform over the entire space $\mathcal{X}$, its unconditional entropy is $\ent(Z) \geq n_0$. Therefore, the mutual information revealed by the protocol is:
$$ \mathrm{I}(Z ; \boldsymbol{\tau}) = \ent(Z) - \ent(Z \mid \boldsymbol{\tau}) \leq n_0 - (n_0 - c_1 - 1 - o(1)) = c_1 + 2.$$
This completes the proof.
\end{proof}


Now we focus on the second part of the proof. Recall that the last Player $k$ (acting as the referee) knows $x_1$ and the full transcript $\boldsymbol{\tau}$. Since $\Pi^*$ solves $\GHM * g$ with an expected error at most $1/16$, the transcript must contain sufficient information to solve the Hidden Matching problem. The following lemma provides the required information lower bound.

\begin{lemma}\label{entropy_loss}
Let $\Pi^{*}$ be any simplified protocol that solves $\GHM * g$ with an overall expected error at most $1/16$. Then, $\mathrm{I}(Z ; \boldsymbol{\tau}) = \Omega(\sqrt{n_0})$.
\end{lemma}

\begin{proof}[Proof of Lemma \ref{entropy_loss}]
Let $\mathcal{M} = \{M_1, \dots, M_m\}$ be the set of $m = n_0/2$ edge-disjoint perfect matchings. Recall that the transcript $\boldsymbol{\tau}$ is independent of the first player's input $x_1$ (which uniformly determines the matching $\boldsymbol{M} \in \mathcal{M}$). Therefore, for any fixed transcript $\tau$ and a specific chosen matching $M \in \mathcal{M}$, the output of the protocol is a constant triple $(i_M, j_M, b_M)$, where $(i_M, j_M) \in M$ is the predicted edge and $b_M \in \{0,1\}$ is the predicted parity. Since the output edge is guaranteed to be in $M$, an error occurs for a fixed $M$ if and only if $b_M \neq Z_{i_M} \oplus Z_{j_M}$.

Let $\varepsilon_\tau$ denote the overall error probability of the protocol conditioned solely on the transcript $\boldsymbol{\tau} = \tau$. Since the overall expected error is $\mathbb{E}_{\boldsymbol{\tau}}[\varepsilon_{\boldsymbol{\tau}}] \leq 1/16$, Markov's inequality implies that $\Pr[\boldsymbol{\tau} \in \mathcal{T}_{\text{good}}] \geq 1/2$, where $\mathcal{T}_{\text{good}}$ is defined as the set of transcripts satisfying $\varepsilon_\tau \leq 1/8$. 

Fix any good transcript $\tau \in \mathcal{T}_{\text{good}}$. Let 
$$\varepsilon_{\tau, M} = \Pr[b_M \neq Z_{i_M} \oplus Z_{j_M} \mid \boldsymbol{\tau} = \tau, \boldsymbol{M} = M]$$
denote the error conditioned on both the transcript $\tau$ and a specific matching $M$. 

Since the expected error over uniformly chosen matchings is $\mathbb{E}_{M \in \mathcal{M}}[\varepsilon_{\tau, M}] = \varepsilon_\tau \leq 1/8$, applying Markov's inequality again yields a subset of matchings $\mathcal{M}'_\tau := \{ M \in \mathcal{M} : \varepsilon_{\tau, M} \leq 1/4 \}$ of size $|\mathcal{M}'_\tau| \geq m/2 = n_0/4$. We construct a simple bipartite graph $G_\tau' = (L, R, E_\tau')$ whose edges are exactly the constant predictions $(i_M, j_M)$ made under each $M \in \mathcal{M}'_\tau$. Because the original matchings are mutually edge-disjoint, $G_\tau'$ contains exactly $|E_\tau'| = |\mathcal{M}'_\tau| \geq n_0/4$ edges.

By Tur\'an's theorem, a simple bipartite graph containing a spanning forest of rank $r$ spans at most $2r$ vertices and thus has at most $r^2$ edges. Consequently, the rank of $G_\tau'$ must be at least $\sqrt{|E_\tau'|} \geq \sqrt{n_0/4}$. Let $A_\tau \subseteq E_\tau'$ be such a spanning forest, which provides $|A_\tau| = \Omega(\sqrt{n_0})$ linearly independent edges.

Let $u_Z, v \in \{0,1\}^{|A_\tau|}$ denote the true parities and the predicted parities on $A_\tau$, respectively. By construction, each edge $e \in A_\tau$ originates from a unique matching in $\mathcal{M}'_\tau$; let $M_e$ denote this matching. Since every matching in $\mathcal{M}'_\tau$ has a conditional error of at most $1/4$, the probability of predicting the wrong parity for edge $e$ is exactly $\varepsilon_{\tau, M_e} \leq 1/4$. By linearity of expectation, the expected Hamming distance between the true and predicted parities on $A_\tau$ is bounded by:
$$ \mathbb{E}_{Z}[ d_H(u_Z, v) \mid \boldsymbol{\tau} = \tau ] = \sum_{e \in A_\tau} \varepsilon_{\tau, M_e} \leq \frac{1}{4} |A_\tau|. $$

Next, we apply the following lemma, a generalization of Fano's inequality due to Bar-Yossef, Jayram, and Kerenidis~\cite{bar2004exponential}, to upper bound the entropy of $u_Z$.

\begin{lemma}[\cite{bar2004exponential}]\label{fano}
Let $\rv{W}$ be a random variable over $\{0,1\}^k$, and suppose there exists a fixed vector $v \in \{0,1\}^k$ such that $\mathbb{E}[d_H(\rv{W}, v)] \leq \epsilon \cdot k$ for some $0 \leq \epsilon \leq 1/2$. Then the entropy of $\rv{W}$ is bounded by,
\[
\ent(\rv{W}) \leq k \cdot H_2(\epsilon),
\]
where $H_2(\cdot)$ is the binary entropy function.
\end{lemma}

Applying the generalized Fano's inequality (Lemma \ref{fano}) with $k = |A_\tau|$ and error bound $1/4$, we have $\ent(u_Z \mid \boldsymbol{\tau}=\tau) \leq |A_\tau| H_2(1/4)$. Furthermore, because the forest $A_\tau$ is acyclic, its parities impose $|A_\tau|$ independent linear constraints on $Z$, leaving exactly $n_0 - |A_\tau|$ bits of uncertainty for any fixed $u_Z$. By the chain rule,
$$ \ent(Z \mid \boldsymbol{\tau}=\tau) \leq \ent(u_Z \mid \boldsymbol{\tau}=\tau) + \ent(Z \mid u_Z, \boldsymbol{\tau}=\tau) \leq |A_\tau| H_2(1/4) + n_0 - |A_\tau| = n_0 - \Omega(\sqrt{n_0}). $$

Taking the expectation over $\boldsymbol{\tau}$ and using the trivial bound $\ent(Z \mid \boldsymbol{\tau}=\tau) \leq n_0$ for bad transcripts, the overall conditional entropy is bounded by:
$$ \ent(Z \mid \boldsymbol{\tau}) \leq \frac{1}{2} (n_0) + \frac{1}{2} \left( n_0 - \Omega(\sqrt{n_0}) \right) = n_0 - \Omega(\sqrt{n_0}). $$

Finally, since the unconditional entropy of the target is $\ent(Z) \geq n_0$, the mutual information revealed by the protocol is precisely lower-bounded:
$$ \mathrm{I}(Z ; \boldsymbol{\tau}) = \ent(Z) - \ent(Z \mid \boldsymbol{\tau}) \geq \Omega(\sqrt{n_0}). $$
This completes the proof.
\end{proof}

\normalem
\bibliographystyle{alpha}
\bibliography{reference.bib}

\appendix
\newpage
\section*{Appendix}
\paragraph{Quantum protocols for $\GHM*g$ \cite{bar2004exponential}:}
We present a quantum protocol for lifted hidden matching problem with communication complexity of $O(\log n)$ qubits. and $x_2,...,x_k \in \{0,1\}^n$ be the first player’s input and $x_1,...,x_{k-1} \in \{0,1\}^n$ be the last player’s input.

Let $z=(z_1,...,z_{n_0}) = g(x_2,...,x_k)$.
\begin{enumerate}
    \item The first player  sends the state \( |\psi\rangle = \frac{1}{\sqrt{n}} \sum_{i=1}^{n_0} (-1)^{z_i} |i\rangle. \)
    
    \item The last player performs a measurement on the state \( |\psi\rangle \) in the orthonormal basis
    \[
    B = \left\{ \frac{1}{\sqrt{2}} (|k\rangle \pm |\ell\rangle) \mid (k, \ell) \in M_{x_1)} \right\}.
    \]
\end{enumerate}

The probability that the outcome of the measurement is a basis state \( \frac{1}{\sqrt{2}} (|k\rangle + |\ell\rangle) \) is
\[
|\langle \psi | \frac{1}{\sqrt{2}} (|k\rangle + |\ell\rangle) \rangle|^2 = \frac{1}{2n} \left( (-1)^{z_k} + (-1)^{z_\ell} \right)^2.
\]

This equals $2/n$ if $z_k \oplus z_\ell = 0$ and $0$ otherwise. Similarly, for the states \( \frac{1}{\sqrt{2}} (|k\rangle - |\ell\rangle) \), we have that
\[
|\langle \psi | \frac{1}{\sqrt{2}} (|k\rangle - |\ell\rangle) \rangle|^2 = 0 \quad \text{if } z_k \oplus z_\ell = 0, \text{ and } \frac{2}{n} \text{ if } z_k \oplus z_\ell = 1.
\]

Hence, if the outcome of the measurement is a state \( \frac{1}{\sqrt{2}} (|k\rangle + |\ell\rangle) \), then the last player knows with certainty that $z_k \oplus z_\ell = 0$ and outputs \( \langle k, \ell, 0 \rangle \). If the outcome is a state \( \frac{1}{\sqrt{2}} (|k\rangle - |\ell\rangle) \), then the last player knows with certainty that $z_k \oplus z_\ell = 1$ and hence outputs \( \langle k, \ell, 1 \rangle \). Note that the measurement depends only on the last player’s input and that the algorithm is $0$-error.

\end{document}